\documentclass[acmtog]{acmart}
\renewcommand\footnotetextcopyrightpermission[1]{}

\setcopyright{none}

\acmDOI{}
\acmISBN{}
\acmConference[Conference]{}{}{}
\acmBooktitle{}

\AtBeginDocument{
  }

 \setcopyright{acmlicensed}

\acmSubmissionID{377}

\usepackage{color}
\definecolor{AAA}{rgb}{1.0, 0.13, 0.32}
\definecolor{BBB}{rgb}{0.2, 0.1, 1}
\definecolor{CCC}{rgb}{0.0, 1, 0}
\definecolor{DDD}{rgb}{0.9, 0, 0.4}

\usepackage{enumitem}
\setlist{nosep}

\usepackage{overpic}
\begin{document}

\title{P2M++: Enhanced Solver for Point-to-Mesh Distance Queries}

\author{Qinghao Guo}
\authornote{Both authors contributed equally to this research.}
\email{1300165109@qq.com}
\orcid{0009-0003-7289-3557}
\affiliation{
  \institution{Shandong University}
  \country{China}
}

\author{Pengfei Wang}
\authornotemark[1]
\email{pengfei1998@foxmail.com}
\orcid{0000-0002-2079-275X}
\affiliation{
  \institution{Shandong University}
  \country{China}
}

\author{Chen Zong}
\email{zongchen@nuaa.edu.cn}
\affiliation{
  \institution{Nanjing University of Aeronautics and Astronautics}
  \country{China}
}

\author{Maodong Pan}
\email{maodong@nuaa.edu.cn}
\affiliation{
  \institution{Nanjing University of Aeronautics and Astronautics}
  \country{China}
}

\author{Shiqing Xin}
\authornote{Corresponding author.}
\email{xinshiqing@sdu.edu.cn}
\affiliation{
  \institution{Shandong University}
  \country{China}
}

\author{Shuangmin Chen}
\email{csmqq@163.com}
\affiliation{
  \institution{Qingdao University of Science and Technology}
  \country{China}
}

\author{Changhe Tu}
\email{chtu@sdu.edu.cn}
\affiliation{
  \institution{Shandong University}
  \country{China}
}

\author{Wenping Wang}
\email{wenping@tamu.edu}
\affiliation{
  \institution{Texas A\&M University}
  \country{USA}
}

\renewcommand{\shortauthors}{Guo et al.}

\begin{abstract}
Point-to-mesh distance queries are fundamental in computer graphics and geometric modeling. While the state-of-the-art \textsc{P2M} method achieves high-speed queries via Voronoi-based localization, it suffers from prohibitive precomputation costs. Its iterative Voronoi sweep for interference detection leads to redundant predicate evaluations and scales poorly on rotationally symmetric structures (e.g., spheres, cones or cylinders), where candidate counts grow quadratically. 

We propose \textsc{P2M++} to address these limitations through three key contributions. First, we adaptively augment the set of mesh vertices with auxiliary sites in regions of high Voronoi vertex density to localize complex interference within minimal spatial regions. Second, we reformulate interference detection as a series of sphere-triangle collision tests centered at Voronoi cell corners, which are efficiently resolved using the base mesh's BVH. Finally, we enhance runtime performance by replacing the standard $kd$-tree search with a faster recursive dynamic programming implementation. Experimental results demonstrate that \textsc{P2M++} is $3\times$--$10\times$ faster than the original \textsc{P2M} during preprocessing and $1.5\times$ faster in queries, with even more pronounced gains on rotationally symmetric geometries.
\end{abstract}



\maketitle

\section{Introduction}

Querying the distance from a query point to a triangle mesh is a fundamental operation in computer graphics and geometric modeling~\cite{10.5120/16754-7073,10.2312/hpg.20191189}. This operation is essential for a wide range of applications, including collision detection in physical simulations~\cite{Teschner2005,10.1145/285857.285860,10.1145/566654.566623}, surface reconstruction from point clouds~\cite{Berger2017,10.5555/1281957.1281965,10.1145/142920.134011}, shape analysis using distance fields~\cite{Jones2006,10.1145/344779.344899}, and path planning in robotics~\cite{Quinlan1994,508439,770022,5152817}.

Despite the widespread use of Bounding Volume Hierarchies (BVH) for retrieving the closest point, the \textsc{P2M} algorithm~\cite{p2m:2023} has recently established a new state-of-the-art for query speed. By locating the mesh triangle whose Voronoi cell contains the query point, \textsc{P2M} effectively bypasses traditional tree-traversal bottlenecks. To avoid the topological complexity of computing true triangle-site Voronoi diagrams, \textsc{P2M} uses the mesh vertex set as proxy sites. It bridges the gap between the nearest vertex and the closest triangle by precomputing ``interference'' relationships. Specifically, if there exists a query point $q$ such that the nearest vertex site is $v$, but the true closest point on the mesh lies on a triangle $f$, then $v$ is said to be an \textit{interceptor} of $f$. Despite its superior query speed, the prohibitive precomputation costs of \textsc{P2M} significantly limit its practical utility.

In this paper, we propose \textsc{P2M++}, an improved framework based on a careful analysis of these inefficiencies. Suppose a mesh vertex $v$ dominates a Voronoi cell $\Omega_v$ within the Voronoi diagram of mesh vertices. \textsc{P2M} identifies interference by sweeping the Voronoi diagram starting from $v$, traversing adjacent cells layer-by-layer to exhaustively find all triangles that could potentially yield a shorter distance than $v$ for any query point $q \in \Omega_v$. This process suffers from two primary bottlenecks. First, the termination criterion is overly conservative: interference detection only stops when an entire layer of Voronoi cells yields no candidates, leading to redundant and costly geometric predicate evaluations. Second, for rotationally symmetric structures such as spheres, cones, or cylinders, the number of interference candidates grows quadratically. This causes both preprocessing and query performance to degrade severely in these common geometric cases.

To address these limitations, we introduce two core improvements. First, our framework adaptively augments the proxy site set with auxiliary points in regions where Voronoi vertices are highly clustered. Unlike the original method, these auxiliary sites are not restricted to the mesh surface, allowing us to localize complex interference within minimal spatial regions. Second, for a mesh vertex $v$ acting as a Voronoi site, we reduce the time-consuming interference detection problem to a series of sphere-triangle collision tests. By checking collisions between the base mesh's BVH and balls centered at the corners of $\Omega_v$---where each ball's radius equals the distance between $v$ and the corresponding corner---we can rapidly identify potential triangle candidates.

For auxiliary (non-vertex) Voronoi sites, we adopt a simplified strategy: we construct a sphere centered at the site with an estimated radius and identify all triangles colliding with this sphere. Importantly, the core strategies of \textsc{P2M++} can be extended to any convex decomposition of the ambient space, provided the decomposition supports an efficient point-location strategy. Finally, we enhance runtime performance by replacing the standard $kd$-tree search with a more efficient recursive dynamic programming implementation.

Our results demonstrate that for common organic models, \textsc{P2M++} is $3\times$--$10\times$ faster than the original \textsc{P2M} during preprocessing and $1.5\times$ faster during queries. These performance gains are even more pronounced for rotationally symmetric structures, where \textsc{P2M} often struggles. 

Our primary contributions are as follows:
\begin{enumerate}[leftmargin=*, nosep]
    \item We present \textsc{P2M++}, an enhanced distance query framework that significantly accelerates both the precomputation and runtime query stages of the current state-of-the-art \textsc{P2M}.
    \item We reformulate interference detection as a series of sphere-triangle collision tests centered at Voronoi cell corners, enabling efficient candidate pruning via the base mesh's BVH.
    \item We introduce an adaptive site augmentation strategy to localize complex interference within minimal spatial regions, specifically addressing bottlenecks caused by geometric symmetry.
    \item We provide a highly optimized implementation that demonstrates superior robustness and performance across diverse public datasets, consistently outperforming both \textsc{P2M} and existing competitors such as FCPW.
\end{enumerate}

\section{Related Work}

Point-to-mesh distance computation is a cornerstone of computer graphics, physical simulation~\cite{10.1016/j.jvcir.2007.01.005,10.1145/566654.566623,10.1145/882262.882358}, and computational geometry~\cite{Botsch2010PolygonMP,10.1109/TVCG.2005.49,10.1007/s10514-013-9327-2}. Existing acceleration algorithms can be broadly categorized into two paradigms: \textit{Hierarchical Spatial Structures} that organize mesh triangles into tree-based indices~\cite{10.1145/3596711.3596791,10.1145/142920.134011}, and \textit{Proxy-based Spatial Precomputation} methods that map query regions to candidate triangle subsets.

\subsection{Hierarchical Spatial Structures}
The most widely adopted strategy utilizes hierarchical structures to organize mesh triangles. These methods accelerate queries via a branch-and-bound traversal, which recursively prunes geometric regions that cannot contain the closest point by maintaining and refining distance bounds during the search.

Classic implementations, such as the Proximity Query Package (PQP)~\cite{Larsen1999FastPQ}, utilize Oriented Bounding Boxes (OBBs)~\cite{BAREQUET200191,10.1145/237170.237244} to tightly wrap geometry, minimizing intersection tests through compact enclosures. Similarly, CGAL~\cite{cgal:pt-t3-25b} employs Axis-Aligned Bounding Boxes (AABBs)~\cite{10.1080/10867651.1997.10487480,LARSSON2006450}, which trade tighter fitting for faster construction and simpler intersection tests, balancing preprocessing efficiency with query performance. In recent years, high-performance libraries have pushed this paradigm by exploiting modern hardware capabilities. Intel Embree~\cite{10.1145/2601097.2601199} leverages SIMD vectorization and wide BVH structures to enable simultaneous testing of multiple bounding volumes. FCPW~\cite{FCPW} further optimizes this approach through vectorized traversal and efficient packet-based query processing, achieving state-of-the-art performance on modern CPU architectures.

\subsection{Proxy-based Spatial Precomputation}
A distinct category of algorithms shifts the focus from hierarchical object organization to precomputing the correspondence between spatial query regions and candidate triangles. These methods identify a small, relevant subset of primitives by pre-calculating ``interception'' or ``influence'' lists.

The \textsc{P2M} (Point-to-Mesh) algorithm~\cite{p2m:2023} is a prominent example of this paradigm. It uses mesh vertices as proxies for the underlying primitives. For a query point $q$, a $kd-tree$ can rapidly retrieve the nearest mesh vertex $v$, implying that $q$ lies within the Voronoi cell $\Omega_v$ of the vertex set. \textsc{P2M} precomputes an \textit{interception table} for each vertex $v$, identifying all edges and faces that could potentially provide a distance shorter than $dist(q, v)$ for any $q \in \Omega_v$. In essence, the precomputation captures the discrepancy between the vertex-based spatial partition and the true primitive-based Voronoi diagram. During the query stage, once the nearest vertex is located, the algorithm only tests the small subset of primitives stored in the corresponding interception table.

Following a similar philosophy, Triangle Influence Supersets~\cite{10.1111:cgf.14861} employ octree-based~\cite{MEAGHER1982129} spatial subdivisions instead of Voronoi regions. For each octree node, a conservative list of candidate triangles is stored based on convex hull influence regions. These interception-based methods achieve significant query speedups over BVH approaches, often by several fold. However, they are characterized by considerably longer preprocessing times and higher memory overhead, as resolving the intrinsic discrepancy between the simplified spatial partition and the complex primitive-based Voronoi boundaries remains a significant challenge.

\section{Methodology}
\subsection{Overview}

Our framework, \textsc{P2M++}, builds upon the foundational concepts of \textsc{P2M}~\cite{p2m:2023}. Similar to its predecessor, \textsc{P2M++} utilizes a set of prescribed proxy points to partition the ambient space and maintains an \textit{interception table} to track true triangle candidates once the nearest proxy point is identified.

\begin{figure}[h]
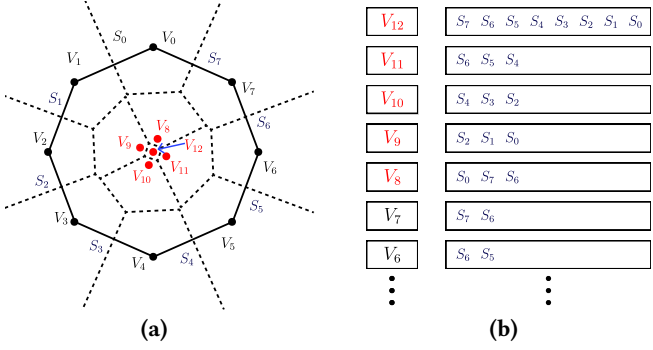

  \centering
  \begin{overpic}[width=\linewidth]{img/fig1.pdf}
    \put(21, -4){\textbf{(a)}} 
    \put(75, -4){\textbf{(b)}}
  \end{overpic}
  \vspace{0.05cm}
  \caption{Overview of the \textsc{P2M++} framework. (a) Augmented Voronoi partitioning combining mesh vertices and auxiliary sites. (b) The interception table explicitly links each Voronoi cell to intercepted geometric primitives (edges in 2D), strictly constraining the search space.}
  \label{fig:Overview:Spatial:partition}
\end{figure}

\textsc{P2M++} introduces three fundamental advancements over the original method. First, while \textsc{P2M} relies solely on the mesh vertex set $V$ to construct the Voronoi diagram $\mathcal{V}_{\text{vertices}}$, \textsc{P2M++} adaptively augments this set with auxiliary sites in regions where Voronoi vertices are highly clustered. This results in an augmented Voronoi diagram $\mathcal{V}_{\text{augmented}}$ that provides a more balanced spatial partition, as illustrated in Fig.~\ref{fig:Overview:Spatial:partition}(a).

Second, the interception table construction is reformulated. \textsc{P2M} identifies interceptors by sweeping $\mathcal{V}_{\text{vertices}}$ for every triangle and edge in the mesh, a process that incurs significant redundant computational overhead. In contrast, \textsc{P2M++} traverses each Voronoi cell in $\mathcal{V}_{\text{augmented}}$ corresponding to a site $v$ and directly retrieves the candidate triangles that identify $v$ as their interceptor. These candidates are defined as triangles that potentially contain the true closest point for at least one query point $q$ within the Voronoi cell of $v$ (see Fig.~\ref{fig:Overview:Spatial:partition}(b)).

Finally, we optimize the query stage by replacing the traditional $kd$-tree search with a recursive Nearest Neighbor Search (NNS) algorithm~\cite{11165079}. Notably, recursive NNS is particularly effective when data points lie on a 2-manifold surface, providing superior query performance compared to standard tree-based approaches.

\begin{figure}[ht]
  \centering
  \hspace{-0.3cm}
  \includegraphics[width=0.5\linewidth]{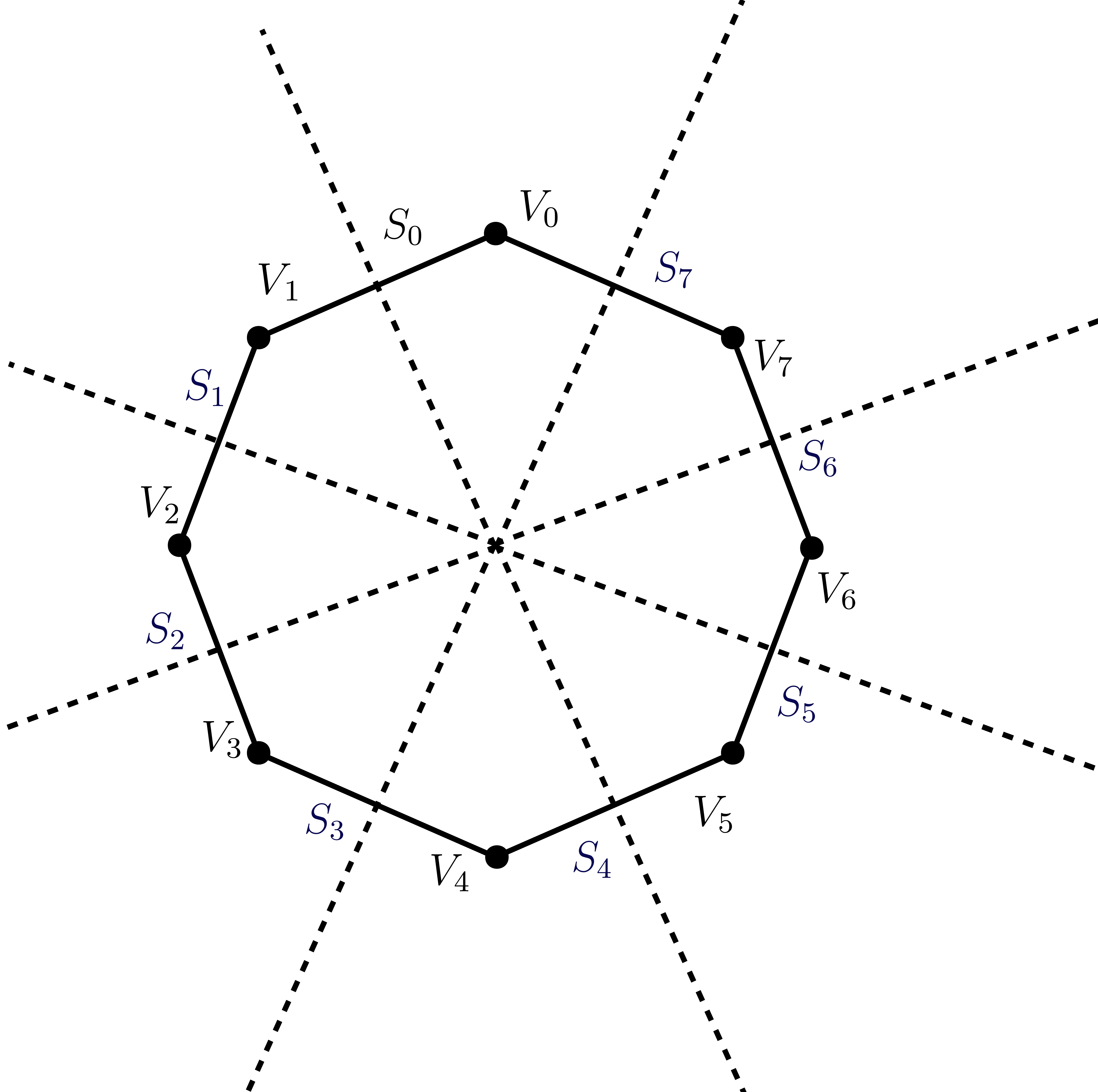}
  \caption{Motivation for site augmentation: In cases of high geometric symmetry, such as a regular $n$-gon approximating a circle, Voronoi vertices cluster at the center, leading to a quadratic explosion in interference relationships.}
  \label{fig:Necessity}
\end{figure}
\subsection{Adaptive Proxy Point Placement}

\paragraph{Motivation for Site Augmentation} 
Consider the approximation of a circle using a regular $n$-gon, as illustrated in Fig.~\ref{fig:Necessity}. As the resolution $n$ increases, the Voronoi vertices concentrate densely at the center of the circle. For a query point $q$ located at this center, any vertex of the polygon may be retrieved as the nearest site, whereas any polygon edge could simultaneously provide the true shortest distance. Consequently, an interference relationship is established between every vertex-edge pair, causing the size of the \textsc{P2M} interception table to grow quadratically, $O(n^2)$, with the resolution $n$. This complexity makes the precomputation and storage requirements for rotationally symmetric or near-symmetric geometries prohibitive, necessitating the introduction of auxiliary sites to localize these high-interference regions and restore computational efficiency.

\paragraph{Ring-based Site Augmentation}
The simplest approach to isolate high-density Voronoi vertex clusters would be to enclose them within a sphere or a convex polytope. However, such a strategy would necessitate a specialized point-location procedure during the query stage. Instead, we exploit a fundamental advantage of the Voronoi diagram: locating a query point within a specific cell is equivalent to a highly efficient Nearest Neighbor Search (NNS). By utilizing additional Voronoi sites to partition the space, we ensure that the query logic remains streamlined and consistent.

We introduce a scale parameter $\tau$ and conceptually partition the ambient space into uniform cells. By traversing the precomputed Voronoi vertices, we identify cells that are densely populated. For each such cell, we insert a central proxy point $p_c$ along with a surrounding set of auxiliary sites $\{p_c^{(i)}\}$ arranged on an enclosing sphere, so that the resulting Voronoi cell of $p_c$ isolates the clustered vertices from the rest of the diagram. When neighboring cells are jointly flagged as dense, overlapping auxiliary sites are removed to keep the partition consistent.

\paragraph{Choice of Parameters}
The cell size, density threshold, and number of auxiliary sites are chosen empirically based on the input mesh. In particular, $\tau$ is set proportional to the local sampling density, and we found the method to be robust within a reasonable range of these settings.

\subsection{Interception Table Construction}

\paragraph{Relaxation of Interference Relationships}
Let $v$ be an arbitrary reference point on the base mesh $\mathcal{M}$ (typically a proxy site), and let $q \in \mathbb{R}^3$ be a query point. A fundamental observation is that any triangle $f \in \mathcal{M}$ capable of providing a distance shorter than the distance to the reference point, $\|q - v\|$, must intersect the open ball $B(q, \|q - v\|)$ centered at $q$. 

When the query point $q$ varies within a convex polyhedral region $\mathcal{P}$, we seek to identify the \textit{compact interception table}, defined as the set of triangles that are truly closest to at least one point in $\mathcal{P}$:
\begin{equation}
\mathcal{L}_{\text{compact}}(v; \mathcal{P}) = \{ f \in \mathcal{M} \mid \exists q \in \mathcal{P}, \text{ s.t. } f = \arg\min_{f' \in \mathcal{M}} d(q, f') \}.
\end{equation}

Computing $\mathcal{L}_{\text{compact}}(v; \mathcal{P})$ directly is computationally challenging because it requires determining if $\mathcal{P}$ intersects the complex bisector surfaces between different triangles of the mesh $\mathcal{M}$.
To maintain efficiency while ensuring completeness, we define the \textit{relaxed interception table} as the set of triangles that could potentially provide a shorter distance than $v$ for at least one $q \in \mathcal{P}$:
\begin{equation}
\mathcal{L}_{\text{relaxed}}(v; \mathcal{P}) = \{ f \in \mathcal{M} \mid \exists q \in \mathcal{P}, \text{ s.t. } d(q, f) < d(q, v) \}.
\end{equation}

For an on-surface proxy point $v \in \mathcal{M}$, the set $\mathcal{L}_{\text{relaxed}}$ naturally contains the triangle(s) incident to $v$. While the region $\mathcal{P}$ contains infinitely many points, rendering an exhaustive search impossible, the following theorem demonstrates that this problem can be reduced to a finite set of discrete tests.

\begin{figure}[htbp]
 \centering
 \includegraphics[width=0.8\linewidth]{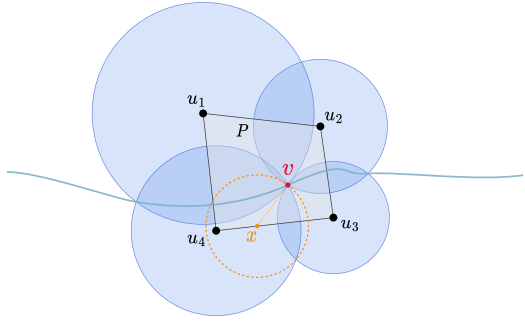}
 \caption{Illustration of the Vertex Sphere Union property. For a convex polygon $\mathcal{P}$ with vertices $\{u_i\}$ and a reference point $v$ on the mesh surface, the union of spheres centered at the vertices, $\bigcup_{u \in \text{vert}(\mathcal{P})} B(u, \|u - v\|)$, completely encloses any sphere $B(x, \|x - v\|)$ centered at any arbitrary point $x$ within $\mathcal{P}$.}
 \label{fig:ver_ball_edge_ball}
\end{figure}

\paragraph{Vertex Sphere Union Theorem}
As illustrated in Fig.~\ref{fig:ver_ball_edge_ball}, identifying $\mathcal{L}_{\text{relaxed}}$ is equivalent to evaluating the union of spheres centered only at the vertices (corners) of $\mathcal{P}$. 

\begin{theorem}[Vertex Sphere Union]
Let $\text{vert}(\mathcal{P})$ denote the set of vertices of a convex polyhedron $\mathcal{P}$. For any reference point $v \in \mathcal{M}$, the continuous union of balls centered at all points $x \in \mathcal{P}$ is contained within the discrete union of balls centered exclusively at its vertices:
\begin{equation}
\bigcup_{x \in \mathcal{P}} B(x, \|x - v\|) = \bigcup_{u \in \text{vert}(\mathcal{P})} B(u, \|u - v\|).
\end{equation}
\end{theorem}

\begin{proof} 
Let $p$ be an arbitrary point in the union $\bigcup_{x \in \mathcal{P}} B(x, \|x - v\|)$. By definition, there exists a point $x \in \mathcal{P}$ such that $p \in B(x, \|x - v\|)$, which implies:
\begin{equation} 
\|p - x\|^2 \leq \|x - v\|^2 \iff \|p - x\|^2 - \|x - v\|^2 \leq 0.
\end{equation}
Define the function $f: \mathbb{R}^3 \to \mathbb{R}$ as $f(y) = \|p - y\|^2 - \|y - v\|^2$. Expanding the terms: 
\begin{align} 
f(y) &= (\|p\|^2 - 2p \cdot y + \|y\|^2) - (\|y\|^2 - 2v \cdot y + \|v\|^2) \\ 
     &= 2y \cdot (v - p) + (\|p\|^2 - \|v\|^2).
\end{align} 
Thus, $f(y)$ is an affine function with respect to $y$. Since $x$ lies within the convex polyhedron $\mathcal{P}$, it can be expressed as a convex combination of its vertices: $x = \sum_i \lambda_i u_i$ with $\lambda_i \geq 0$ and $\sum_i \lambda_i = 1$. By the linearity of $f$: 
\begin{equation} 
f(x) = f\left(\sum_i \lambda_i u_i\right) = \sum_i \lambda_i f(u_i).
\end{equation}
Since $p \in B(x, \|x - v\|)$, we have $f(x) \leq 0$. A convex combination of values can be non-positive only if at least one of the values is non-positive. Therefore, there exists some vertex $u_k$ such that $f(u_k) \leq 0$, which implies $\|p - u_k\| \leq \|u_k - v\|$, i.e., $p \in B(u_k, \|u_k - v\|)$. This proves $\bigcup_{x \in \mathcal{P}} B(x, \|x - v\|) \subseteq \bigcup_{u \in \text{vert}(\mathcal{P})} B(u, \|u - v\|)$. The reverse inclusion is trivial.
\end{proof}

\paragraph{Implementation Details}
When the reference site $v$ is an on-surface proxy point, we reduce the interference detection to identifying triangles intersected by the union of vertex-centered spheres $$\bigcup_{u \in \text{vert}(\mathcal{P})} B(u, \|u - v\|),$$where $\mathcal{P}$ is the Voronoi cell of $v$. This sphere-mesh intersection problem is efficiently resolved using the Bounding Volume Hierarchy (BVH) of the base mesh $\mathcal{M}$.

For off-surface proxy points added during site augmentation, since the site $s$ does not lie on the mesh surface, we choose the reference point as the projection $\operatorname{proj}(s, \mathcal{M})$ of $s$ onto the mesh. However, their Voronoi cells often have many vertices, making individual sphere construction per vertex computationally expensive. To reduce preprocessing cost, we construct a single conservative sphere centered at the site $s$ with radius:
\begin{equation}
r = \max_{u \in \text{vert}(\mathcal{P})} \left( \|s - u\| + \|u - \operatorname{proj}(s, \mathcal{M})\| \right)
\end{equation}
By triangle inequality, this sphere contains the union of all vertex-centered spheres constructed with respect to the reference point $\operatorname{proj}(s, \mathcal{M})$, ensuring correctness. This reduces the number of sphere-mesh intersection tests from $O(|\text{vert}(\mathcal{P})|)$ to $O(1)$ per site.

\subsection{Distance Query Stage}

Given a query point $q \in \mathbb{R}^3$, the distance computation proceeds in two primary phases. First, we identify the nearest proxy point $v$ via a Nearest Neighbor Search (NNS) and retrieve its corresponding precomputed interception table $\mathcal{L}_{\mathcal{P}}$. We then evaluate the distance from $q$ to each triangle in the table, returning the minimum value. To further enhance query performance, we introduce the following two strategies.

\paragraph{Recursive Nearest Neighbor Search}
While our framework incorporates off-surface auxiliary proxy points, the vast majority of the proxy sites lie on the 2-manifold surface $\mathcal{M}$. The recursive NNS algorithm~\cite{11165079} has been demonstrated to offer superior query performance when data points are distributed on a manifold. 
Therefore, we replace the traditional $kd$-tree search with this recursive NNS algorithm to better exploit the local geometric coherence of the proxy sites.

\paragraph{Bounding Sphere Pruning}
We precompute a bounding sphere $B(c, r)$ for each triangle in the mesh $\mathcal{M}$. During the traversal of candidate triangles associated with site $v$, we apply a pruning test: the exact point-to-triangle distance is computed only if the lower bound of the distance, $\|q - c\| - r$, is smaller than the current minimum distance $d_{\min}$. 
Otherwise, the triangle is discarded. This filtering technique effectively skips redundant geometric calculations with a minimal computational overhead for each test.

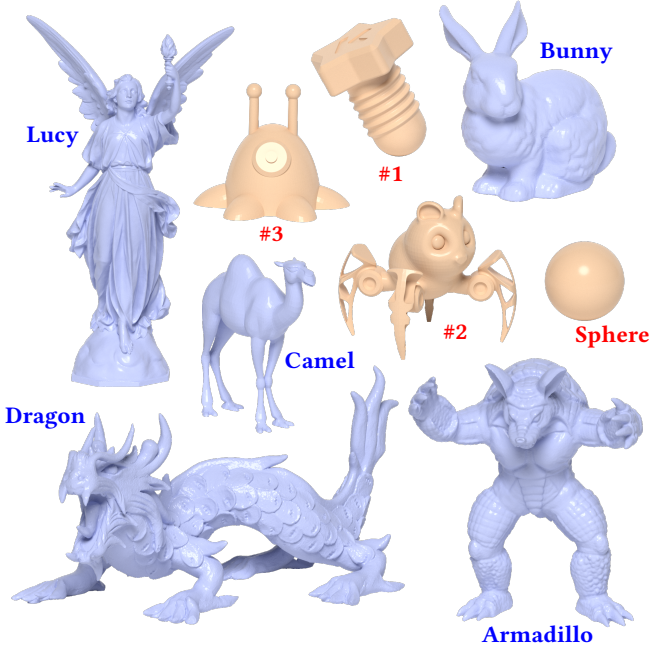
\begin{figure}[!h]
  \centering
  \begin{overpic}[width=\linewidth]{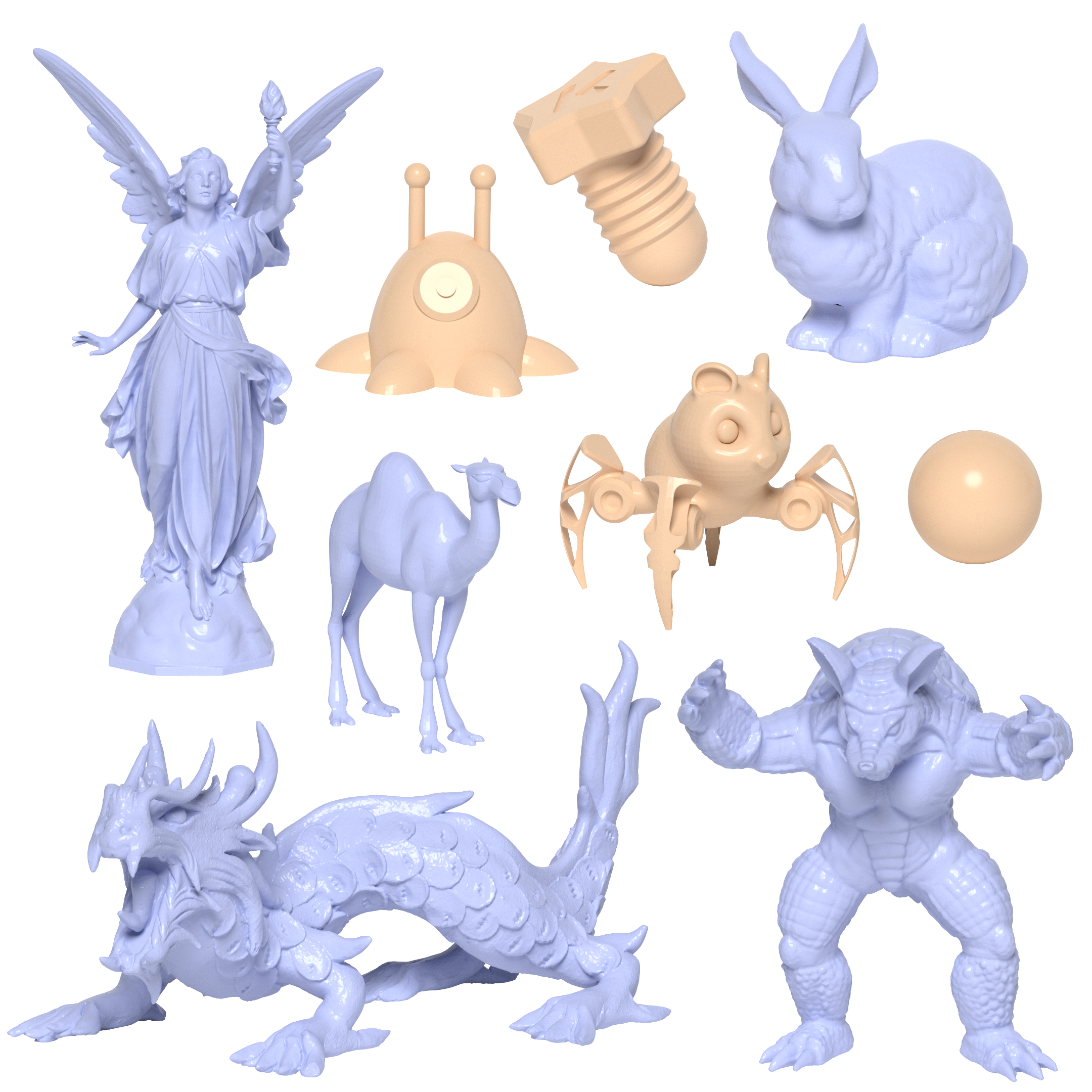}
    \put(57.5,70){\color{red}\textbf{\#1}}
    \put(67.5,45.5){\color{red}\textbf{\#2}}
    \put(39,60.6){\color{red}\textbf{\#3}}
    \put(82.5,89){\color{blue}\textbf{Bunny}}
    \put(-0.3,32.4){\color{blue}\textbf{Dragon}}
    \put(73.5,-1.5){\color{blue}\textbf{Armadillo}}
    \put(3,76){\color{blue}\textbf{Lucy}}
    \put(43,41){\color{blue}\textbf{Camel}}
    \put(88,45){\color{red}\textbf{Sphere}}
  \end{overpic}
  \caption{Models used for quantitative comparison in the tables: classic free-form geometries (blue) and rotationally symmetric models (yellow).}
  \label{fig:model}
  \vspace{-2mm}
\end{figure}

\section{Evaluation}

We implement our method in C++ and conduct all experiments on a Mac mini equipped with an Apple M4 CPU and 16GB RAM. Our implementation leverages CGAL~\cite{10.1145/1653771.1653865} for constructing the Delaunay triangulation, Embree~\cite{10.1145/2601097.2601199} for accelerating sphere-triangle intersection tests during interception table construction, and DPNN~\cite{11165079} library for nearest neighbor queries during the query stage.

We compare our method against two state-of-the-art approaches: the original \textsc{P2M} algorithm~\cite{p2m:2023} and the FCPW library~\cite{FCPW} with vectorization enabled. Following the experimental setup in \textsc{P2M}, we generate one million query points uniformly distributed within a region 10$\times$ the axis-aligned bounding box of each test model. We evaluate our method on models from the ABC dataset~\cite{Koch_2019_CVPR}, which provides a diverse collection of CAD geometries with varying complexity, as well as several classic benchmark models commonly used in the graphics community. To provide detailed quantitative comparisons, we select representative models spanning different geometric characteristics—including classic free-form geometries and rotationally symmetric structures—as visualized in Fig.~\ref{fig:model}. For clarity, we refer to the classic models by their original names and denote the rotationally symmetric models as \#1, \#2, and \#3.

\begin{table}[]
\caption{Statistics of the interception tables across various models. Max and Avg denote the maximum and arithmetic mean of the interception list length per site, respectively. Queried denotes the average number of triangle-distance comparisons performed per query.}
\label{table:inter}
\resizebox{0.99\columnwidth}{!}{%
\begin{tabular}{cccccc|cccc}
\toprule
                             & Armadillo & Bunny  & Camel  & Dragon & Lucy   & \#1 & \#2 & \#3 & Sphere \\ \hline
\multicolumn{1}{c|}{faces}   & 99976     & 69451  & 19510  & 249882 & 525814 & 48434    & 142562  & 15902   & 20480  \\ \hline
\multicolumn{1}{c|}{Max}     & 371       & 609    & 216    & 598    & 799    & 43186    & 2485    & 15902   & 20480  \\
\multicolumn{1}{c|}{Avg}     & 63.639    & 62.732 & 66.921 & 62.868 & 63.930 & 225.86   & 151.59  & 243.91  & 229.15 \\
\multicolumn{1}{c|}{Queried} & 72.314    & 59.601 & 57.383 & 71.850 & 99.710 & 49.36    & 104.81  & 46.19   & 25.51  \\
\bottomrule
\end{tabular}
}
\end{table}

\subsection{Interception List Analysis}

\begin{figure}[htbp]
  \centering
  \includegraphics[width=\linewidth]{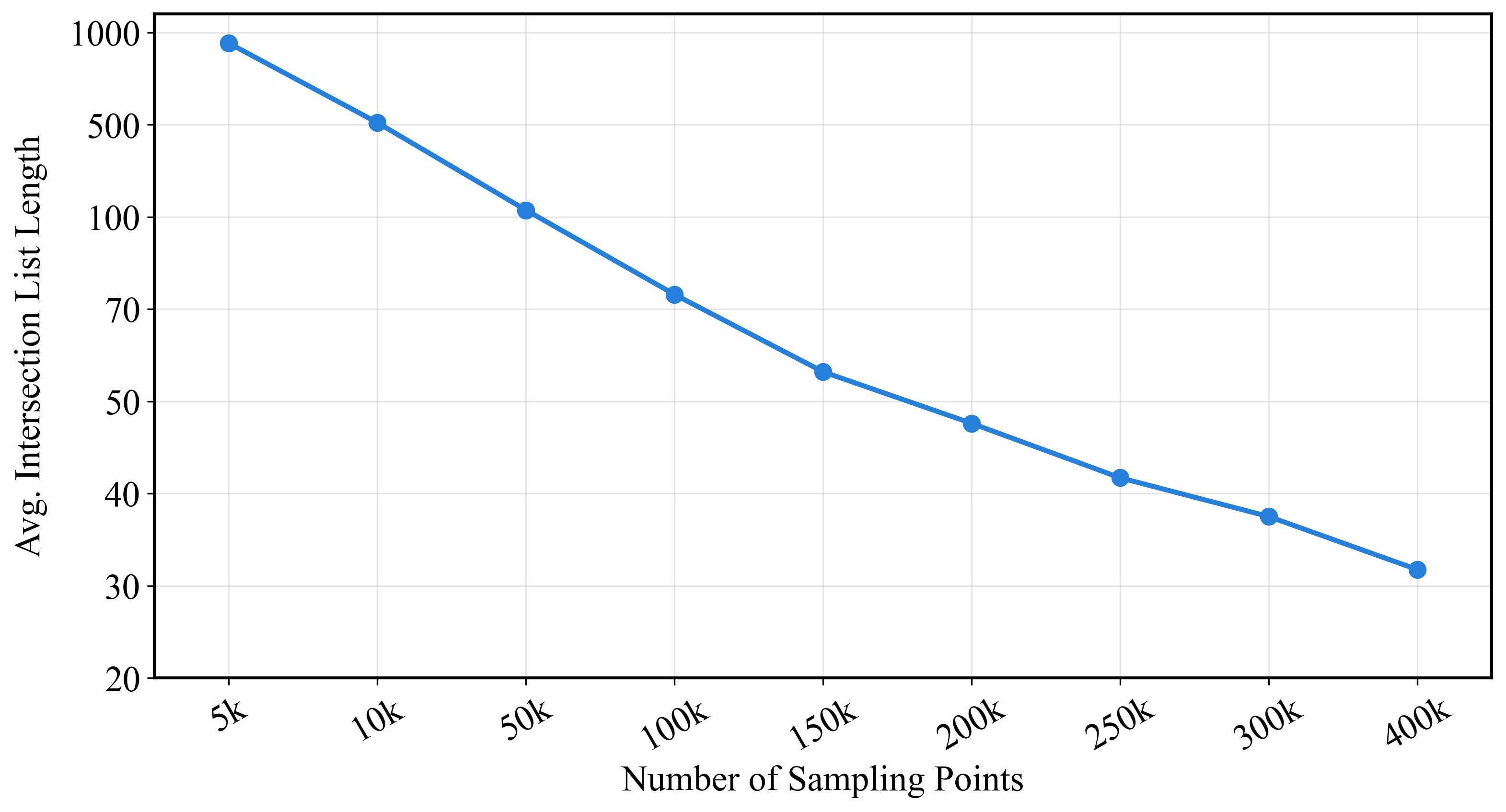}
  \caption{Interception list length w.r.t the number of sampling points on the Dragon model.}
  \label{fig:Intersection_list_length}
  \vspace{-2mm}
\end{figure}

In our method, each Voronoi site maintains an interception table, and the length of this table directly impacts query performance. We analyze three metrics: the average interception table length per site, the maximum interception table length, and the average number of triangle-distance comparisons per query over one million queries. The third metric—which we term the query-weighted average—more accurately reflects actual query performance than the per-site average. This is because Voronoi cells have varying volumes: a site with a long interception table but a small cell volume contributes minimally to overall query cost, as query points rarely fall within that region.

Table~\ref{table:inter} presents these statistics for several models. For conventional geometries (e.g., Bunny, Dragon), the maximum interception table length remains significantly smaller than the total number of faces, indicating that no single cell requires checking the majority of the mesh. Moreover, the per-site average closely matches the query-weighted average, demonstrating balanced spatial distribution.

In contrast, rotationally symmetric models (e.g., Sphere) exhibit a stark discrepancy between maximum and average lengths. The maximum length approaches O(n), reflecting the presence of geometrically unstable regions near symmetry centers where slight perturbations drastically change the nearest triangle. Naively, this would cause the per-site average to be prohibitively high. However, our adaptive site insertion strategy confines these unstable regions to small, localized cells. As a result, despite high maximum lengths, the query-weighted average remains comparable to conventional models, as query points rarely encounter these problematic regions. This demonstrates the effectiveness of our adaptive refinement in mitigating the pathological behavior of symmetric geometries.

Furthermore, taking the classic Dragon model (250k faces) as an example, Fig.~\ref{fig:Intersection_list_length} illustrates how the per-site average interception table length varies as we increase the number of sampling sites. As the site count grows, the average table length exhibits a steady decline. This demonstrates that denser spatial partitioning effectively reduces the search scope within each cell. For well-behaved geometries, increasing the number of sites through adaptive refinement uniformly improves the compactness of interception tables, leading to more efficient queries.

\begin{table}[]
\centering
\caption{
Query time comparison (in $\mu s$) for FCPW, \textsc{P2M}, and our method across different models.
The breakdown of our method into Cell Locate and List Search stages is also provided. The \textbf{\underline{best}} and \underline{second-best} results are highlighted.}
\vspace{-2mm}
\label{table:query_classical_models}
\resizebox{0.99\columnwidth}{!}{%
\begin{tabular}{cccccc|cccc}
\toprule
 & Armadillo & Bunny & Camel & Dragon & Lucy & \#1 & \#2 & \#3 & Sphere \\ \hline
\multicolumn{1}{c|}{faces} & 99976 & 69451 & 19510 & 249882 & 525814 & 48434 & 142562 & 15902 & 20480 \\ \hline
\multicolumn{1}{c|}{FCPW} & 1.842 & \underline{1.074} & 1.051 & 1.162 & 1.450 & \underline{1.338} & 3.386 & 1.546 & \underline{3.286} \\
\multicolumn{1}{c|}{\textsc{P2M}} & \underline{1.099} & 1.345 & \textbf{\underline{0.474}} & \underline{0.755} & \underline{0.861} & 1.394 & \underline{1.424} & \underline{0.703} & 39.122 \\
\multicolumn{1}{c|}{Ours} & \textbf{\underline{0.885}} & \textbf{\underline{0.417}} & \underline{0.534} & \textbf{\underline{0.671}} & \textbf{\underline{0.803}} & \textbf{\underline{0.393}} & \textbf{\underline{0.721}} & \textbf{\underline{0.418}} & \textbf{\underline{0.482}} \\ \hline
\multicolumn{1}{c|}{Cell Locate} & 0.313 & 0.167 & 0.207 & 0.345 & 0.353 & 0.218 & 0.253 & 0.171 & 0.241 \\
\multicolumn{1}{c|}{List Search} & 0.553 & 0.236 & 0.305 & 0.314 & 0.441 & 0.167 & 0.451 & 0.222 & 0.221 \\ \bottomrule
\end{tabular}
}
\vspace{-2mm}
\end{table}

\subsection{Query Performance}
Our query process consists of two main steps: locating the nearest Voronoi site and traversing the corresponding interception table to find the closest triangle. Table~\ref{table:query_classical_models} compares query performance across several standard models and provides the time breakdown for each stage of our method. Our query time achieves a 1–2$\times$ speedup over \textsc{P2M} and is significantly faster than FCPW (1–5$\times$). Notably, \textsc{P2M} exhibits severe performance degradation on models with globally high rotational symmetry (e.g., the Sphere) due to the breakdown of its interception table strategy. In such cases, our method is up to 50$\times$ faster thanks to our adaptive site insertion strategy, which effectively prevents pathological behavior and also secures a 2$\times$+ speedup on geometries with local rotational symmetry. In terms of time breakdown, nearest site lookup accounts for approximately 40\% of query time on average, while interception table traversal constitutes the remaining 60\%.

Fig.~\ref{fig:query_time_compare_on_ABC_dataset} demonstrates query performance across the ABC dataset, showing consistent performance advantages of our method across diverse geometries.

\begin{figure}[htbp]
  \centering
  \includegraphics[width=\linewidth]{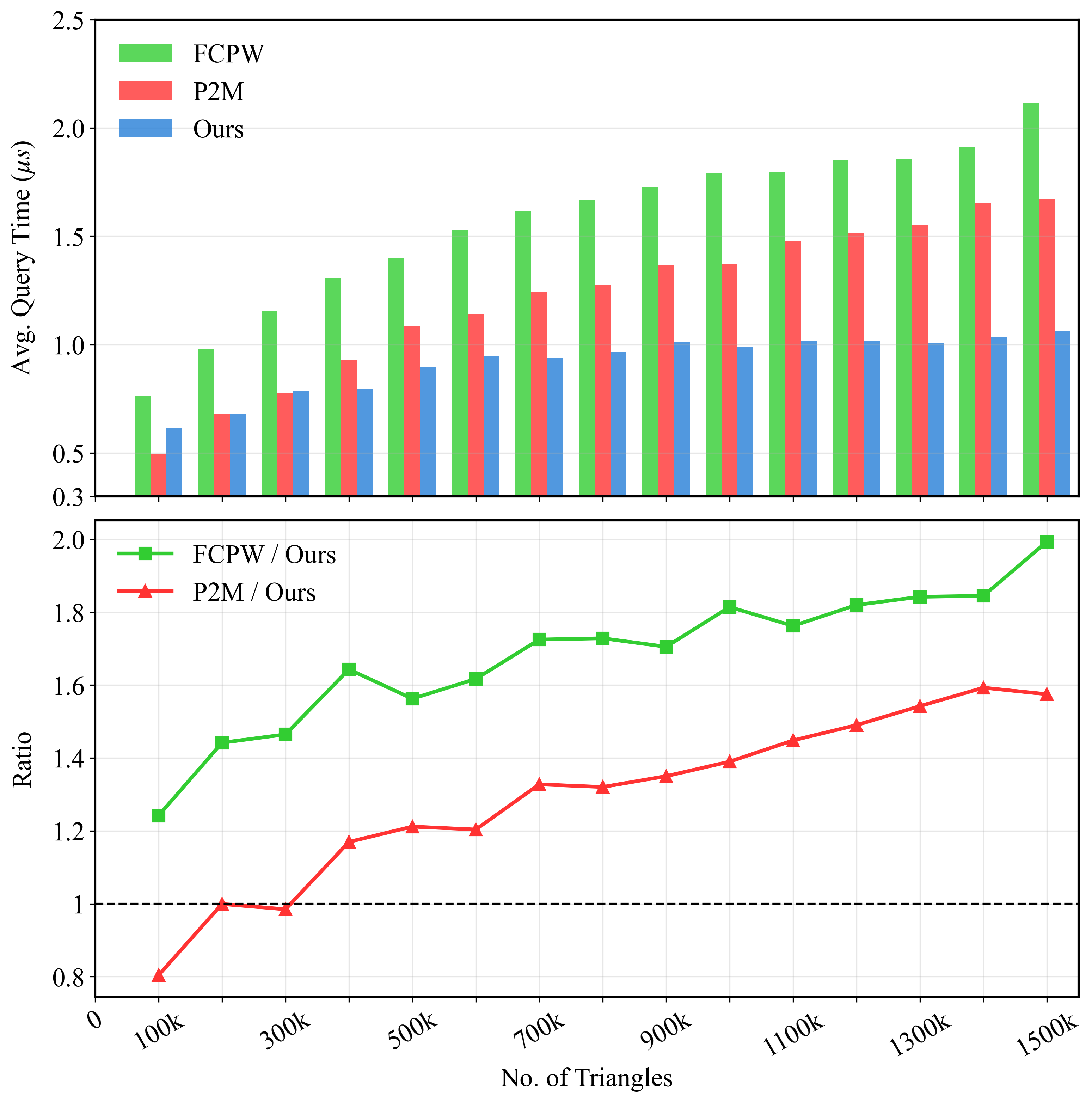}
  \caption{Comparison about query performance on the Dragon model with varying resolutions. Top: the average timing cost per query for P2M, FCPW and ours. Bottom: the comparison about the query cost among P2M, FCPW and ours.}
  \label{fig:query_time_compare}
\end{figure}

Fig.~\ref{fig:query_time_compare} examines scaling behavior on the Dragon model as triangle count increases, illustrating the comparative query performance of different methods across multiple resolution scales.

\begin{figure}[!h]
  \centering
  \includegraphics[width=\linewidth]{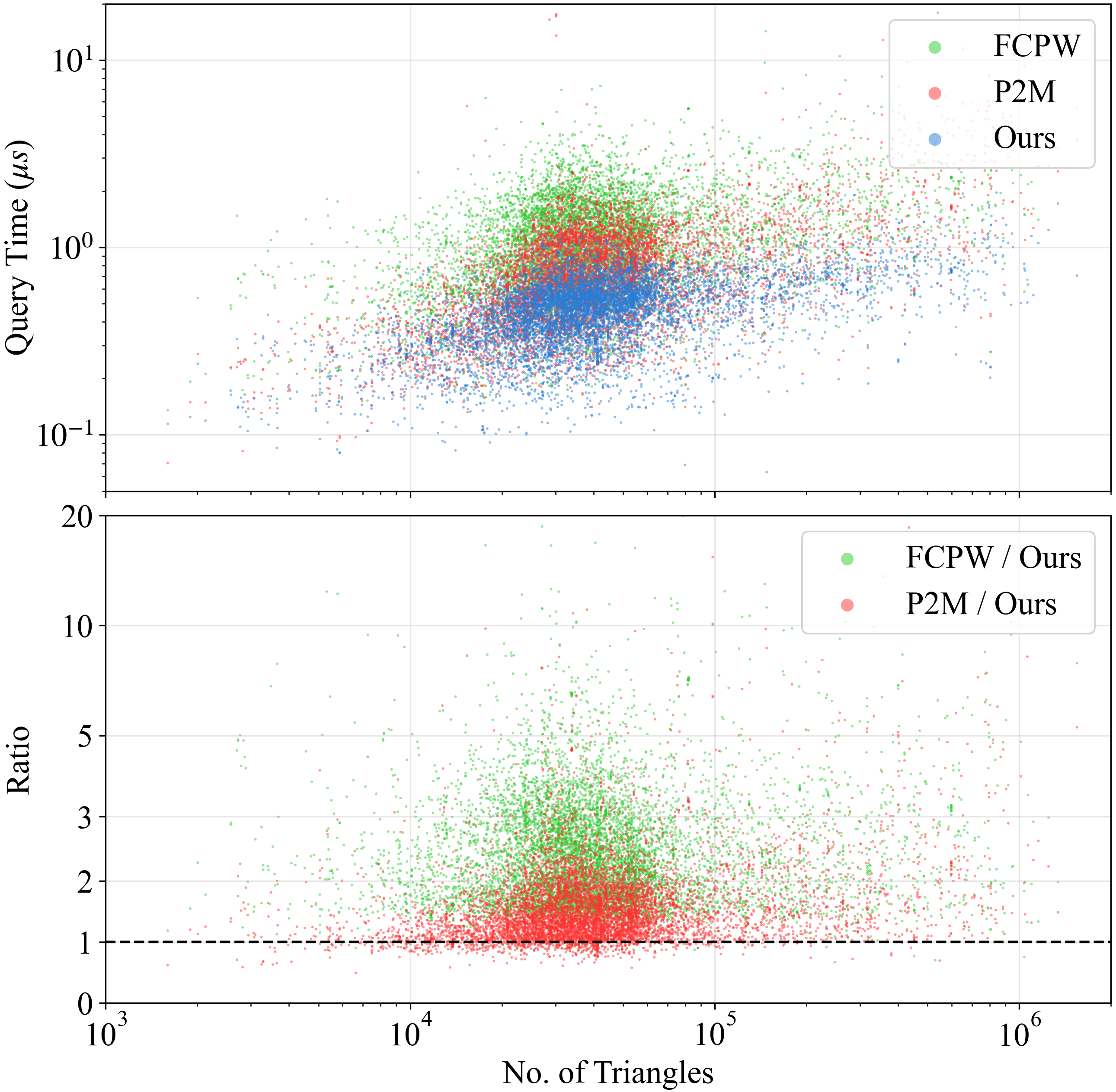}
  \caption{Query time comparison for \textsc{P2M}, FCPW, and our method on the ABC dataset.}
  \label{fig:query_time_compare_on_ABC_dataset}
\end{figure}

\begin{table}[]
\caption{Preprocessing time (in seconds) on classic and rotationally symmetric models. The table shows total time for FCPW, \textsc{P2M}, and our method, along with the breakdown of our three preprocessing stages.}
\vspace{-2mm}
\label{table:PrecessingTimeOnClassicalModels}
\resizebox{0.99\columnwidth}{!}{%
\begin{tabular}{cccccc|cccc}
\toprule
 & Armadillo & Bunny & Camel & Dragon & Lucy & \#1 & \#2 & \#3 & Sphere \\ \hline
\multicolumn{1}{c|}{faces} & 99976 & 69451 & 19510 & 249882 & 525814 & 48434 & 142562 & 15902 & 20480 \\ \hline
\multicolumn{1}{c|}{FCPW} & 0.026 & 0.017 & 0.005 & 0.068 & 0.152 & 0.012 & 0.040 & 0.004 & 0.005 \\
\multicolumn{1}{c|}{\textsc{P2M}} & 5.007 & 3.300 & 0.863 & 12.354 & 27.346 & 23.568 & 14.609 & 4.234 & 420.533 \\
\multicolumn{1}{c|}{Ours} & 0.694 & 0.534 & 0.137 & 1.739 & 3.895 & 0.795 & 1.548 & 0.352 & 0.791 \\ \hline
\multicolumn{1}{c|}{Voronoi Build} & 0.305 & 0.215 & 0.061 & 0.774 & 1.682 & 0.243 & 0.508 & 0.070 & 0.073 \\
\multicolumn{1}{c|}{Optimize} & 0.075 & 0.062 & 0.020 & 0.164 & 0.380 & 0.226 & 0.140 & 0.124 & 0.162 \\
\multicolumn{1}{c|}{List Build} & 0.313 & 0.256 & 0.056 & 0.801 & 1.831 & 0.325 & 0.899 & 0.157 & 0.556 \\ \bottomrule
\end{tabular}
}
\vspace{-2mm}
\end{table}

\subsection{Preprocessing Time}

Our preprocessing involves three main stages: dpnn and Delaunay triangulation construction (performed together as both rely on the Delaunay triangulation), adaptive site insertion, and interception table construction.

Table~\ref{table:PrecessingTimeOnClassicalModels} compares the total preprocessing time of our method with \textsc{P2M} and FCPW on several models, along with the time breakdown for each stage of our method. Both our method and \textsc{P2M} are slower than FCPW in preprocessing, as FCPW only requires BVH construction. However, we achieve significant improvements over \textsc{P2M} while maintaining comparable query performance. On conventional geometries, our method achieves {3–10$\times$} speedup in preprocessing. The advantage becomes even more pronounced on highly rotational symmetric models, where \textsc{P2M}'s preprocessing time grows prohibitively expensive, making it impractical for real-world applications, while our method maintains reasonable preprocessing costs.

Fig.~\ref{fig:preprocess_time_compare_on_ABC_dataset} demonstrates the preprocessing performance across the ABC dataset, which contains diverse CAD geometries with varying complexity and structural characteristics. Our method demonstrates consistent preprocessing speedups over \textsc{P2M} across the majority of models in this diverse collection.

\begin{figure}[htbp]
  \centering
  \includegraphics[width=\linewidth]{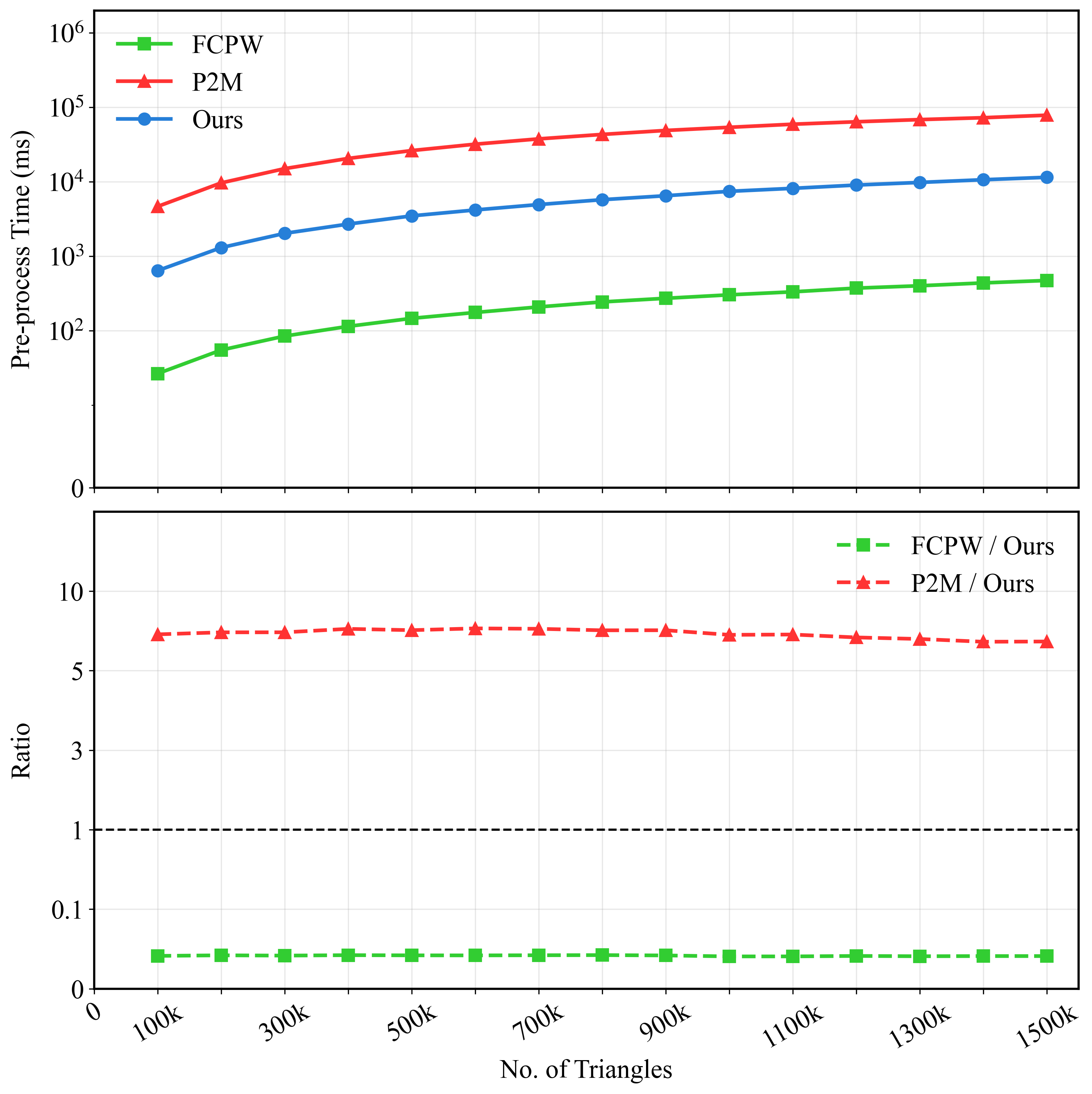}
  \caption{The preprocessing cost on the Dragon model with varying resolutions. Top: our preprocessing cost v.s. the mesh resolution. Bottom: the comparison about the preprocessing cost among P2M, FCPW and ours.}
  \label{fig:preprocess_dragon}
\end{figure}

Fig.~\ref{fig:preprocess_dragon} examines scaling behavior on the Dragon model at multiple resolutions. Our preprocessing time grows gradually with increasing triangle count and maintains a stable performance advantage over \textsc{P2M} across all tested scales, demonstrating the effectiveness of our BVH-accelerated sphere-mesh intersection approach for interception table construction.

\begin{figure}[!h]
  \centering
  \includegraphics[width=\linewidth]{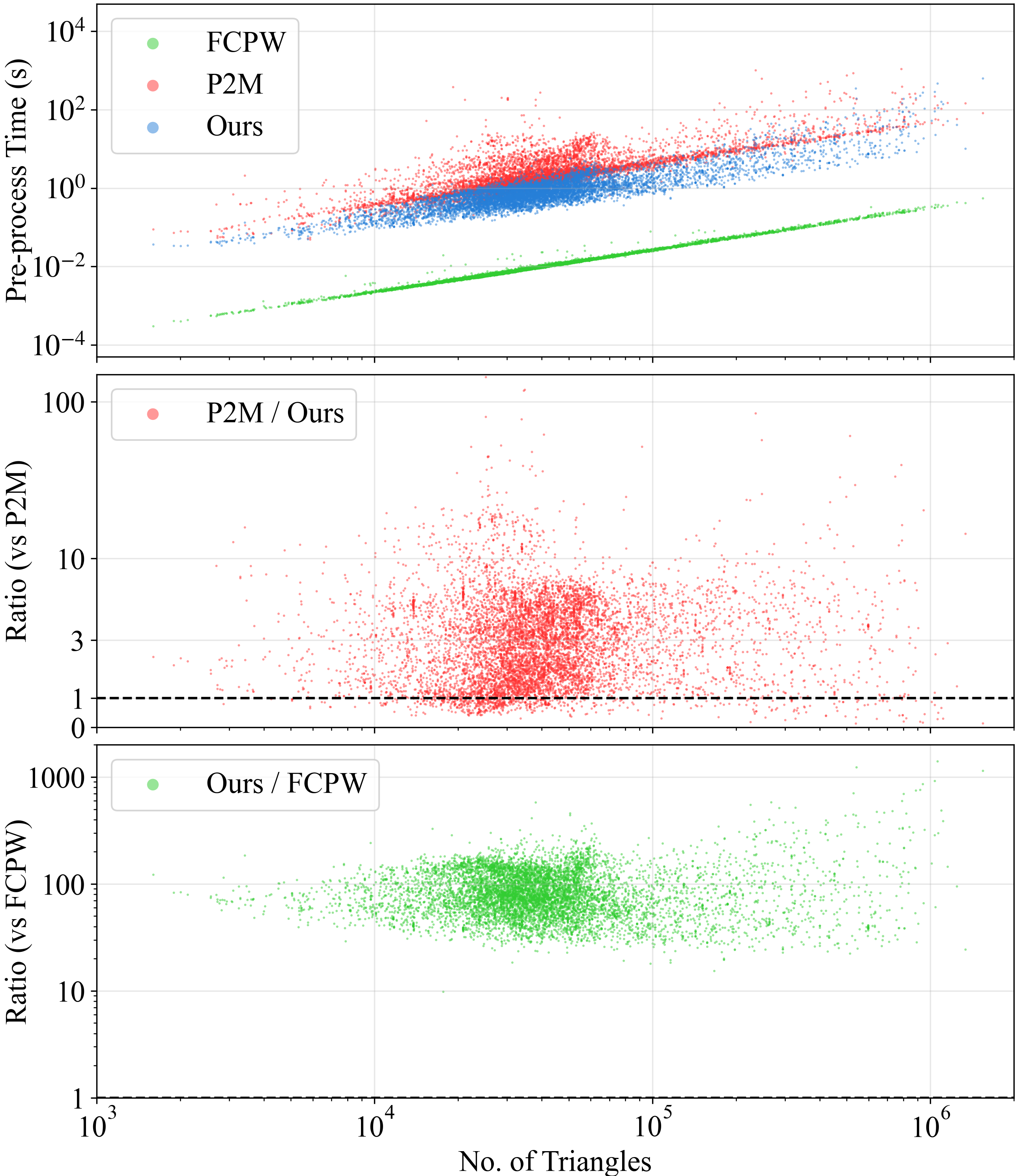}
  \caption{Preprocessing time comparison for \textsc{P2M}, FCPW, and our method on the ABC dataset.}
  \label{fig:preprocess_time_compare_on_ABC_dataset}
  \vspace{-2mm}
\end{figure}

\begin{table}[h]
\caption{Ablation study statistics comparing the method \textbf{without (w/o)} and \textbf{with (w/)} the adaptive site insertion strategy. Statistics include preprocessing time (Pre.), query time, and interception list lengths: maximum (Max), average (Avg), and query-weighted average (Queried).}

\label{tab:ablation_study}
\resizebox{0.99\columnwidth}{!}{%
\setlength{\aboverulesep}{0pt} 
\setlength{\belowrulesep}{0pt} 
\begin{tabular}{l|cc|cc|cc|cc|cc}
\toprule
 & \multicolumn{2}{c|}{Pre. (ms)} & \multicolumn{2}{c|}{Query ($\mu$s)} & \multicolumn{2}{c|}{Max Len.} & \multicolumn{2}{c|}{Avg Len.} & \multicolumn{2}{c}{Queried Len.} \\ 
\cmidrule(lr){2-3} \cmidrule(lr){4-5} \cmidrule(lr){6-7} \cmidrule(lr){8-9} \cmidrule(lr){10-11}
 & w/o & w/ & w/o & w/ & w/o & w/ & w/o & w/ & w/o & w/ \\ \midrule
Sphere & 10064.6 & 791.5 & 15.79 & 0.48 & 20480 & 20480 & 20052.8 & 229.2 & 20070.6 & 25.5 \\
Dragon & 1676.1 & 1739.1 & 0.67 & 0.67 & 598 & 598 & 62.9 & 62.9 & 71.9 & 71.9 \\
\#1 & 2787.8 & 794.5 & 3.08 & 0.39 & 5210 & 43186 & 650.5 & 225.7 & 3335.1 & 49.4 \\ 
\#3 & 454.8 & 351.7 & 2.24 & 0.42 & 2337 & 15902 & 638.9 & 243.9 & 1570.4 & 46.2 \\
\bottomrule
\end{tabular}
}
\vspace{-1mm}
\end{table}

\subsection{Ablation Study}

To verify the effectiveness of our adaptive site insertion strategy, we compare our full method against a configuration excluding this component. Table~\ref{tab:ablation_study} reports the preprocessing time, query time, and key statistics of the interception list across representative models.

The results, summarized in Table~\ref{tab:ablation_study}, highlight the impact of our adaptive strategy. The most significant gains are observed on the \textit{Sphere} model, where our method achieves a 10$\times$ speedup in preprocessing and a 30$\times$ speedup in query performance. This performance leap is driven by a drastic reduction in the search space: the average interception list length drops from $\approx$ 20k to 229.2, and the query-weighted average length plummets from $\approx$ 20k to just 25.5. These quantitative results confirm that our strategy effectively resolves the pathological clustering of Voronoi vertices caused by the sphere's rotational symmetry, which otherwise forces the non-adaptive solver to traverse excessively large cells.

The benefits of our approach extend to irregular models exhibiting local rotational symmetry, such as \textit{\#1} and \textit{\#3}. For these complex geometries, the adaptive strategy yields substantial reductions in both preprocessing overhead and query latency. These efficiency gains are directly attributed to the effective minimization of the query-weighted average length, confirming that our method successfully resolves Voronoi vertex clustering in high-curvature regions. This result underscores the versatility of our solver, demonstrating its ability to optimize spatial partitioning even for arbitrary shapes lacking perfect geometric symmetry.

Finally, as exemplified by the \textit{Dragon} model, the performance metrics remain virtually identical when employing the adaptive strategy. This confirms the robustness of our approach: the mechanism remains inactive on geometries with uniform vertex distribution, ensuring that no unnecessary computational overhead is introduced for general geometries.

\section{Limitations and Future Work}

While our proposed method achieves state-of-the-art query performance and significantly reduces preprocessing costs compared to \textsc{P2M}, it still incurs a higher preprocessing overhead than lightweight approaches like FCPW. In scenarios involving small-scale query batches where this initialization cost cannot be effectively amortized, FCPW may yield a superior total runtime.

In future work, we aim to address these bottlenecks by investigating more efficient geometric predicates for sphere-triangle intersection tests to further accelerate the preprocessing phase. We also plan to explore algorithmic optimizations to enhance both construction and query throughput. Furthermore, we intend to extend our current solver to support mesh-to-mesh distance queries, enabling efficient proximity computations between complex geometric models. Ultimately, these performance gains and functional extensions will extend the applicability of our framework to a broader spectrum of downstream tasks and practical scenarios.

\section{Conclusion}

In this paper, we presented \textsc{P2M++}, a robust and exact point-to-mesh distance query framework built upon an enhanced Voronoi partitioning paradigm. Our approach systematically addresses the fundamental computational bottlenecks of the original \textsc{P2M} algorithm by introducing two primary technical innovations. First, we developed an adaptive site insertion strategy that monitors the density of Voronoi vertices to strategically place auxiliary sites in geometrically unstable regions. This mechanism yields a more balanced spatial decomposition and effectively prevents the quadratic explosion of interception table lengths in rotationally symmetric geometries. Second, we reformulated the interception table construction as a series of sphere-mesh intersection tests centered at Voronoi cell corners. This formulation naturally admits acceleration via Bounding Volume Hierarchies (BVH), significantly streamlining the precomputation phase.

Extensive evaluation across diverse datasets demonstrates that \textsc{P2M++} achieves a $3\times$ to $10\times$ speedup in preprocessing over the original \textsc{P2M} while consistently delivering superior query performance. Notably, our method exhibits strong robustness on rotationally symmetric geometries, such as spheres and cylinders, where prior Voronoi-based solvers often suffer severe performance degradation.

\clearpage

\bibliographystyle{ACM-Reference-Format}
\bibliography{sample-base}

\clearpage

\appendix

\end{document}